\documentclass[a4paper,onecolumn,10pt]{amsart}

\usepackage{amsmath}
\usepackage{amssymb}  
\usepackage{amsthm}
\usepackage{amsfonts}
\usepackage{graphicx}
\usepackage{cancel}
\usepackage{bbm}
\usepackage{url}
\usepackage{enumerate} 
\usepackage{ upgreek }
\usepackage{dsfont}
\usepackage{color}
\usepackage{subcaption}

\usepackage{makecell}
\usepackage{diagbox}
\usepackage{wrapfig}
\usepackage{rotating}
\usepackage{tabularx}

\usepackage{hyperref}
\hypersetup{
    colorlinks=true,
    linkcolor=blue,
    citecolor=red,
    filecolor=magenta,      
    urlcolor=cyan,
}

\newcommand{\ket}[1]{| #1 \rangle}

\newcommand{\bra}[1]{\langle #1 |}

\newcommand{\proj}[2]{| #1 \rangle\!\langle #2 |}

\newcommand{\id}{\ensuremath{\mathds{1}}}

\newcommand{\cB}{\mathcal{B}}

\newcommand{\cU}{\mathcal{U}}

\def\beq{\begin{equation}}
\def\eeq{\end{equation}}
\def\bq{\begin{quote}}
\def\eq{\end{quote}}
\def\ben{\begin{enumerate}}
\def\een{\end{enumerate}}
\def\bit{\begin{itemize}}
\def\eit{\end{itemize}}

\def\ra{\rightarrow}

\def\lb{\left(}
\def\rb{\right)}
\def\lset{\lbrace}
\def\rset{\rbrace}

\def\r|{\right|}
\def\lbr{\left[}
\def\rbr{\right]}
\def\ident{\textnormal{id}}
\def\one{\id}

\newcommand\C{\mathbbm{C}}

\newcommand\R{\mathbbm{R}}
\newcommand\N{\mathbbm{N}}
\newcommand\M{\mathcal{M}}
\newcommand\D{\mathcal{D}}

\DeclareMathOperator{\Tr}{Tr}

\DeclareMathOperator{\Pasym}{P_{\text{asym}}}
\DeclareMathOperator{\Psym}{P_{\text{sym}}}

\theoremstyle{plain}
\newtheorem{thm}{Theorem}[section]
\newtheorem{lem}[thm]{Lemma}

\newtheorem{defn}[thm]{Definition}

\theoremstyle{definition}

\begin{document}
\title{{On the monotonicity of a quantum optimal transport cost}}
\author{Alexander M\"uller-Hermes}
\address{\small{Department of Mathematics, University of Oslo, P.O. box 1053, Blindern, 0316 Oslo, Norway}}
\email{muellerh@math.uio.no}

\begin{abstract}
We show that the quantum generalization of the $2$-Wasserstein distance proposed by Chakrabarti et al. is not monotone under partial traces. This disproves a recent conjecture by Friedland et al. Finally, we propose a stabilized version of the original definition, which we show to be monotone under the application of general quantum channels. 
\end{abstract}

\maketitle
\date{\today}

\section{Introduction}

For every dimension $d\in\N$ we define the symmetric subspace 
\[
\C^d\vee \C^d = \lset \ket{\psi}\in \C^d\otimes \C^d ~:~\mathbbm{F}_d \ket{\psi} = \ket{\psi}\rset ,
\]
and the antisymmetric subspace
\[
\C^d\wedge \C^d = \lset \ket{\psi}\in \C^d\otimes \C^d ~:~\mathbbm{F}_d \ket{\psi} = -\ket{\psi}\rset ,
\]
where we used the flip operator $\mathbbm{F}_d:\C^d\otimes \C^d\ra \C^d\otimes \C^d$ defined by
\[
\mathbbm{F}\ket{a}\otimes \ket{b} = \ket{b}\otimes \ket{a} ,
\]
and extended linearly. The symmetric subspace and antisymmetric subspace are subspaces of $\C^d\otimes \C^d$ and we denote by $\Psym(d)$ and $\Pasym(d)$ their respective orthogonal projections. In~\cite{zhou2022quantum} these projectors where used to study quantum versions of classical earth mover's distances. Along similar lines, a quantum generalization of the classical Wasserstein distance of order $2$ with respect to the Hamming distance (see~\cite{villani2009optimal} for details) was proposed in~\cite{chakrabarti2019quantum}: 
 
\begin{defn}\label{defn:Wasserstein}
For quantum states $\rho,\sigma\in \D\lb \C^d\rb$ the \emph{quantum optimal transport cost} is given by
\[
T(\rho,\sigma) = \min_{\tau_{AB}\in \D\lb \C^d\otimes \C^d\rb} \Tr\lbr \tau_{AB}\Pasym(d) \rbr ,
\]
where the minimum is over all quantum states $\tau_{AB}\in \D\lb \C^d\otimes \C^d\rb$ with marginals $\tau_A=\rho$ and $\tau_B=\sigma$. The quantity 
\[
W(\rho,\sigma) = \sqrt{T\lb \rho,\sigma\rb} ,
\]
is called the \emph{quantum Wasserstein semi-distance} of order $2$ between $\rho$ and $\sigma$.
\end{defn}

Many properties of these quantities where discussed in~\cite{cole2021quantum,friedland2022quantum}, and it is conjectured that the quantum Wasserstein semi-distance of order $2$ satisfies the triangle inequality, which would make it a genuine distance. It should be pointed out that there are other proposals for quantum generalizations of Wasserstein distances and optimal transport costs including definitions for orders different than $2$. See for example~\cite{zyczkowski1998monge,carlen2014analog,golse2018wave,de2021quantum2,de2021quantum,cole2021quantum} and references therein. In this article, we will only discuss the quantum optimal transport cost and quantum Wasserstein semi-distance of order $2$ from Definition \ref{defn:Wasserstein}. For brevity, we will refer to these quantities as \emph{the} quantum optimal transport cost and \emph{the} quantum Wasserstein semi-distance although we do not claim that these are the ``correct'' generalizations in any sense. In fact, we will even argue that these definitions have shortcomings, when used as distance measures in a context of processing quantum information. 

It was asked in~\cite{bistron2022monotonicity} whether the quantum optimal transport cost (and hence the quantum Wasserstein semi-distance) is monotone under the application of quantum channels, i.e., whether the inequality 
\[
T(\phi(\rho),\phi(\sigma))\leq T(\rho,\sigma),
\]
holds for all pairs of quantum states $\rho,\sigma\in \D(\C^d)$ and any quantum channel $\phi:\M_d\ra \M_{d'}$. It was observed (see~\cite{bistron2022monotonicity}) that this monotonicity holds when $d=d'=2$ and when $\phi$ is a mixed unitary quantum channel. Based on numerical evidence, it was then conjectured in~\cite{friedland2022quantum} that monotonicity holds in general. Here, we will show that this conjecture is false. The quantum optimal transport cost is \emph{not} monotone under general quantum channels, and in particular it fails to be monotone for arbitrary pairs of quantum states $\rho,\sigma$ when $\phi$ is a partial trace. This article is structured as follows:
\begin{itemize}
\item In Section \ref{sec:Props} we will review some useful properties of the quantum optimal transport cost. 
\item In Section \ref{sec:mon} we will show that the quantum optimal transport cost is not monotone under partial traces.
\item In Section \ref{sec:stab} we will introduce a stabilized version of the quantum optimal transport cost. We will show that this quantity is monotone under the application of quantum channels and that it shares many properties of the original definition. 
\end{itemize}

\section{Properties of the quantum optimal transport cost}\label{sec:Props}

To prove properties of the quantum optimal transport cost it will be useful to consider an alternative ordering of the tensor factors when having tensor product input states. We will use the notation
\[
\Pasym(d_1\otimes d_2) = \frac{1}{2}\lb \one_{d_1}\otimes \one_{d_1}\otimes \one_{d_2}\otimes \one_{d_2} - \mathbbm{F}_{d_1}\otimes \mathbbm{F}_{d_2}\rb ,
\]
to denote the projector onto the antisymmetric subspace $(\C^{d_1}\otimes \C^{d_2})\wedge (\C^{d_1}\otimes \C^{d_2})$, but where the tensor factors are reshuffled so that $\Pasym(d_1\otimes d_2)$ acts on $\C^{d_1}\otimes\C^{d_1}\otimes \C^{d_2}\otimes\C^{d_2}$. We will often use letters `A' and `B' to clarify the situation: We would, for example, denote by $A_1$ and $B_1$ the two tensor factors with dimension $d_1$, and by $A_2$ and $B_2$ the two tensor factors with dimension $d_2$. The antisymmetric projector $\Pasym(d_1\otimes d_2)$ is defined with respect to the exchange of the $A$-labeled and the $B$-labeled tensor factors, where we use the ordering $A_1B_1A_2B_2$ instead of $A_1A_2B_1B_2$. 

For quantum states $\rho_{1},\sigma_{1}\in \D\lb \C^{d_1}\rb$ and $\rho_{2},\sigma_{2}\in \D\lb \C^{d_2}\rb$ we have
\[
T(\rho_{1}\otimes \rho_{2},\sigma_{1}\otimes \sigma_{2}) = \min_{\tau_{A_1 B_1A_2 B_2}} \Tr\lbr\tau_{A_1 B_1A_2 B_2} \Pasym(d_1\otimes d_2) \rbr ,
\]
where the minimum is over all quantum states
\[
\tau_{A_1 B_1A_2 B_2}\in \D\lb \C^{d_1}\otimes\C^{d_1}\otimes\C^{d_2}\otimes \C^{d_2}\rb ,
\]
such that the marginals are
\[
\tau_{A_1A_2} = \rho_{1}\otimes \rho_{2} ,
\]
and 
\[
\tau_{B_1B_2} = \sigma_{1}\otimes \sigma_{2} .
\]
Finally, it is useful to note the identity
\begin{equation}\label{equ:Key}
\Pasym(d_1\otimes d_2) = \Pasym(d_1) \otimes \Psym(d_2) + \Psym(d_1)\otimes \Pasym(d_2).
\end{equation}

\subsection{Formulation as a semidefinite problem}

Clearly, the quantum optimal transport cost $T(\rho,\sigma)$ is a semidefinite program and we can compute its dual as
\[
T^*(\rho,\sigma) = \sup_{E,F} \Tr\lbr E\rho\rbr + \Tr\lbr F \sigma\rbr ,
\]
for $\rho,\sigma\in \D\lb \C^{d}\rb$, where the supremum goes over $E,F\in \M(\C^d)$ such that 
\begin{equation}\label{equ:EFCondition}
E\otimes \one_d + \one_d\otimes F \leq \Pasym(d) .
\end{equation}
It is clear that 
\begin{equation}\label{equ:orderSDPs}
T^*(\rho,\sigma) \leq T(\rho,\sigma),
\end{equation}
since 
\[
\Tr\lbr E\rho\rbr + \Tr\lbr F \sigma\rbr = \Tr\lbr \lb E\otimes \one_d + \one_d\otimes F\rb \tau_{AB}\rbr \leq \Tr\lbr \tau_{AB}\Pasym(d)\rbr ,
\]
for any quantum state $\tau_{AB}$ with marginals $\tau_A=\rho$ and $\tau_B=\sigma$ and any $E,F\in \M(\C^d)$ satisfying \eqref{equ:EFCondition}. Moreover, as shown in~\cite{chakrabarti2019quantum}, strong duality holds for the pair of semidefinite programs, and we have   
\[
T^*(\rho,\sigma) = T(\rho,\sigma) ,
\]
for all $\rho,\sigma\in \D\lb \C^{d}\rb$.

\subsection{Monotonicity under tensoring with a state}

We start with a theorem proved in \cite{friedland2022quantum} showing that the quantum optimal transport cost between two quantum states decreases under tensoring with the another quantum state:

\begin{thm}[\cite{friedland2022quantum}]\label{thm:MonUnderTensor}
For quantum states $\rho,\sigma\in D(\C^{d_1})$ and any quantum state $\gamma\in \D\lb \C^{d_2}\rb$ we have
\[
T(\rho\otimes \gamma,\sigma\otimes \gamma)\leq T(\rho,\sigma).
\]
\end{thm}

For the reader's convenience we include a proof of this result: 

\begin{proof}
The quantum optimal transport cost $T$ is jointly convex (see~\cite[Proposition 2.2]{cole2021quantum}) and it is enough to prove the inequality for any pure state
\[
\gamma=\proj{\psi}{\psi}\in \D\lb \C^{d_2}\rb .
\]
Given a quantum state $\tau_{A_1B_1}\in \D(\C^{d_1}\otimes \C^{d_1})$ with marginals $\tau_{A_1}=\rho$ and $\tau_{B_1}=\sigma$, we note that $\tau_{A_1B_1A_2B_2} = \tau_{A_1B_1}\otimes \proj{\psi}{\psi}\otimes \proj{\psi}{\psi}$ is a quantum state with marginals $\rho\otimes \proj{\psi}{\psi}$ and $\sigma\otimes \proj{\psi}{\psi}$. We find that 
\begin{align*}
T&(\rho\otimes \proj{\psi}{\psi},\sigma\otimes \proj{\psi}{\psi}) \\
&\leq \Tr\lbr \tau_{A_1B_1A_2B_2} \Pasym(d_1\otimes d_2)\rbr \\
&= \Tr\lbr (\tau_{A_1B_1}\otimes \proj{\psi}{\psi}\otimes \proj{\psi}{\psi})(\Pasym(d_1) \otimes \Psym(d_2) + \Psym(d_1)\otimes \Pasym(d_2))\rbr \\
&=\Tr\lbr \tau_{A_1B_1}\Pasym(d_1)\rbr,
\end{align*}
since $\ket{\psi}\otimes \ket{\psi}\in \C^{d_2}\vee \C^{d_2}$. This shows that 
\[
T(\rho\otimes \proj{\psi}{\psi},\sigma\otimes \proj{\psi}{\psi}) \leq T(\rho,\sigma),
\]
and the proof is finished.
\end{proof}

The main result of our article is to show that the inequality in the previous theorem is, in general, not an equality, and there are quantum states where it is strict. This shows, in particular, that the quantum optimal transport cost (and hence the quantum Wasserstein semi-distance) is not monotone under partial traces.

\section{The quantum optimal transport cost is not monotone}\label{sec:mon}

We start with an elementary lemma:

\begin{lem}
For $E,F\in \M(\C^{d_1})$ the following are equivalent:
\begin{enumerate}
\item We have 
\[
E\otimes \one_{d_1} \otimes \one_{d_2}\otimes \one_{d_2} + \one_{d_1}\otimes F \otimes \one_{d_2}\otimes \one_{d_2} \leq \Pasym(d_1\otimes d_2).
\]
\item We have both
\[
E\otimes \one_{d_1} + \one_{d_1}\otimes F \leq \Pasym(d_1) ,
\]
and
\[
E\otimes \one_{d_1} + \one_{d_1}\otimes F \leq \Psym(d_1) .
\]
\end{enumerate}
\end{lem}

\begin{proof}
With \eqref{equ:Key} and since 
\[
\one_{d_2}\otimes \one_{d_2} = \Pasym(d_2) + \Psym(d_2) ,
\]
we find that 
\begin{align*}
&\Pasym(d_1\otimes d_2)- E\otimes \one_{d_1} \otimes \one_{d_2}\otimes \one_{d_2} - \one_{d_1}\otimes F \otimes \one_{d_2}\otimes \one_{d_2} \\
&=(\Pasym(d_1)-E\otimes \one_{d_1} - \one_{d_1}\otimes F) \otimes \Psym(d_2)\\
&\quad\quad\quad + (\Psym(d_1)-E\otimes \one_{d_1} - \one_{d_1}\otimes F)\otimes \Pasym(d_2).
\end{align*}
Since $\Psym(d_2)$ and $\Pasym(d_2)$ are self-adjoint projections onto orthogonal subspaces the statement of the theorem follows.
\end{proof}

To show that the quantum optimal transport cost is not monotone under quantum channels, we will need Hermitian matrices $E,F\in \M\lb \C^d\rb$ satisfying 
\begin{equation}\label{equ:1}
E\otimes \one_d + \one_d\otimes F \leq \Pasym(d) ,
\end{equation}
but also
\begin{equation}\label{equ:2}
E\otimes \one_d + \one_d\otimes F \nleq \Psym(d) .
\end{equation}
Such examples can be found numerically for $d\geq 4$. For instance, we can choose
\[
E = \begin{pmatrix} -1.37 & 0 & 0 & 0 \\
0 & 0.02 & 0 & 0 \\
0 & 0 & 0.17 & 0 \\
0 & 0 & 0 & 0.26 
\end{pmatrix}
\]
and 
\[
F = \begin{pmatrix} 0.1165 & -0.02 + 0.01i & 0.03-0.05i & -0.04-0.05i \\
-0.02-0.01i & -0.0935 & 0.02i & 0.16-0.11i \\
0.03+0.05i & -0.02i & -0.2335 & 0.06+0.11i \\
-0.04 + 0.05i & 0.16 +0.11i & 0.06-0.11i & -1.1435
\end{pmatrix}.
\]
By embedding these examples into matrices of larger dimensions, it is straightforward to show that similar examples exist for any $d\geq 4$. Indeed, for $k\geq 1$ and $\alpha\in \R^+$ consider the Hermitian matrices $E',F'\in \M\lb \C^{d+k}\rb$ given by
\[
E' = \begin{pmatrix} E & 0 \\ 0 & -\alpha\one_k\end{pmatrix} \quad\text{ and }\quad F' = \begin{pmatrix} F & 0 \\ 0 & -\alpha\one_k\end{pmatrix},
\]
where $E,F\in \M\lb\C^d\rb$ are Hermitian matrices satisfying \eqref{equ:1} and \eqref{equ:2}. If $\alpha$ is chosen large enough, then it is straightforward to verify that
\[
E'\otimes \one_{d+k} + \one_{d+k}\otimes F' \leq \Pasym(d+k) ,
\]
and of course we have
\[
E'\otimes \one_{d+k} + \one_{d+k}\otimes F' \nleq \Psym(d+k) .
\]
While this construction yields examples for any $d\geq 4$, we do not know whether such examples exist for $d=3$. We can now prove our main result:

\begin{thm}\label{thm:ViolMon}
For any $d\geq 4$ there are quantum states $\rho,\sigma\in \D\lb \C^d\rb$ such that 
\[
T(\rho,\sigma) > T(\rho\otimes \frac{\one_2}{2},\sigma\otimes \frac{\one_2}{2}).
\]
In particular, the quantum optimal transport cost $T$ and the quantum Wasserstein semi-distance $W$ are not monotone under partial traces over qubit systems.
\end{thm}

\begin{proof}
Let $E,F\in \M(\C^d)$ be such that 
\[
E\otimes \one_d + \one_d\otimes F \leq \Pasym(d) ,
\]
but also
\[
E\otimes \one_d + \one_d\otimes F \nleq \Psym(d) .
\]
Such examples exist by the discussion preceeding the statement of the theorem. Consider the quantum state
\[
\tau_{A_1B_1A_2B_2} = \proj{\psi}{\psi}_{A_1B_1}\otimes \frac{\Pasym(2)}{\Tr\lbr \Pasym(2)\rbr},
\]
for some pure quantum state $\ket{\psi}\in \C^d\otimes \C^d$ satisfying 
\[
\bra{\psi}(E\otimes \one_d + \one_d\otimes F)\ket{\psi} > \bra{\psi}\Psym(d) \ket{\psi}.
\]
We set $\rho = \Tr_{B_1}\lbr \proj{\psi}{\psi}_{A_1B_1}\rbr$ and $\sigma=\Tr_{A_1}\lbr \proj{\psi}{\psi}_{A_1B_1}\rbr$ and note that the quantum state $\tau_{A_1B_1A_2B_2}$ from above has marginals
\[
\tau_{A_1A_2} = \rho\otimes \frac{\one_2}{2} \quad\text{ and }\quad \tau_{B_1B_2}=\sigma\otimes \frac{\one_2}{2} .
\]
We will now show that 
\[
T(\rho,\sigma) > T(\rho\otimes \frac{\one_2}{2},\sigma\otimes \frac{\one_2}{2}).
\]
For this compute
\begin{align*}
T(\rho\otimes \frac{\one_2}{2},\sigma\otimes \frac{\one_2}{2}) &\leq \Tr\lbr \tau_{A_1B_1A_2B_2} \Pasym(d\otimes 2)\rbr \\
&= \Tr\lbr \lb \proj{\psi}{\psi}_{A_1B_1}\otimes \frac{\Pasym(2)}{\Tr\lbr \Pasym(2)\rbr}\rb \Pasym(d\otimes 2)\rbr \\
&= \bra{\psi}\Psym(d) \ket{\psi} \Tr\lbr \frac{\Pasym(2)}{\Tr\lbr \Pasym(2)\rbr} \rbr \\
&< \bra{\psi}(E\otimes \one_d + \one_d\otimes F)\ket{\psi} \\
&= \Tr\lbr E \rho\rbr + \Tr\lbr F\sigma\rbr \\
&\leq T^*(\rho,\sigma) \\
&\leq T(\rho,\sigma),
\end{align*}
where the third line follows from \eqref{equ:Key} and the last step follows from \eqref{equ:orderSDPs}. 
\end{proof}

To gain a better understanding for how the quantum optimal transport cost can fail to be monotone under partial traces, we can use a symmetrization argument. Recall the $UU$-twirling channel $\Sigma_d:\M_d\otimes \M_d \ra \M_d\otimes \M_d$ defined as
\[
\Sigma_d(X) = \int_{\cU_{d}} \lb U\otimes U\rb X \lb U^\dagger\otimes U^\dagger\rb \text{d}U,
\]
where the integral is with respect to the Haar measure on group $\cU_{d}$ of $d\times d$ unitary matrices. Using Schur-Weyl duality~\cite{werner1989quantum}, it can be shown that 
\begin{equation}\label{equ:Twirl}
\Sigma_d(X) = \Tr\lbr X \Psym(d)\rbr \frac{\Psym(d)}{\Tr\lbr \Psym(d)\rbr} + \Tr\lbr X \Pasym(d)\rbr \frac{\Pasym(d)}{\Tr\lbr \Pasym(d)\rbr} .
\end{equation}
With this identity we can show the following theorem:

\begin{thm}\label{thm:MaxMonViol}
For quantum states $\rho_1,\sigma_1\in \D(\C^{d_1})$ and $\rho_2,\sigma_2\in \D(\C^{d_2})$ we have
\[
T(\rho_1\otimes \rho_2,\sigma_1\otimes \sigma_2) \geq T\lb \rho_1\otimes \frac{\one_2}{2}, \sigma_1\otimes \frac{\one_2}{2}\rb .
\]
\end{thm}
\begin{proof}
Consider a quantum state $\tau_{A_1B_1A_2B_2}\in \D(\C^{d_1}\otimes \C^{d_1} \otimes \C^{d_2}\otimes \C^{d_2})$ such that 
\[
\tau_{A_1A_2} = \rho_1\otimes \rho_2 \quad\text{ and }\quad \tau_{B_1B_2} = \sigma_1\otimes \sigma_{2} ,
\]
and such that 
\[
T\lb \rho_1\otimes \rho_2 , \sigma_1\otimes \sigma_2\rb = \Tr\lbr\tau_{A_1B_1A_2B_2}\Pasym(d_1\otimes d_2)\rbr .
\] 
By \eqref{equ:Key} we have 
\[
(\ident_{d_1}\otimes \ident_{d_1}\otimes \Sigma_{d_2})\lb \Pasym(d_1\otimes d_2)\rb = \Pasym(d_1\otimes d_2).
\]
Using \eqref{equ:Twirl} and that the twirl $\Sigma_{d_2}$ is self-adjoint with respect to the Hilbert-Schmidt inner product, we find that 
\[
\Tr\lbr\tau_{A_1B_1A_2B_2}\Pasym(d_1\otimes d_2)\rbr = \Tr\lbr\hat{\tau}_{A_1B_1A_2B_2}\Pasym(d_1\otimes d_2)\rbr ,
\]
where $\hat{\tau}_{A_1B_1A_2B_2}\in \D(\C^{d_1}\otimes \C^{d_1} \otimes \C^{d_2}\otimes \C^{d_2})$ is a quantum state of the form
\[
\hat{\tau}_{A_1B_1A_2B_2} = X_{A_1B_1}\otimes \frac{\Psym(d_2)}{\Tr\lbr \Psym(d_2)\rbr} + Y_{A_1B_1}\otimes \frac{\Pasym(d_2)}{\Tr\lbr \Pasym(d_2)\rbr} ,
\]
with positive operators $X_{A_1B_1},Y_{A_1B_1}\in \cB\lb \C^{d_1}\otimes \C^{d_1}\rb^+$ such that 
\[
X_{A_1B_1} + Y_{A_1B_1} = \tau_{A_1B_1}.
\]
Finally, observe that 
\begin{align*}
\Tr\lbr\hat{\tau}_{A_1B_1A_2B_2}\Pasym(d_1\otimes d_2)\rbr &= \Tr\lbr X_{A_1B_1} \Pasym(d_1)\rbr + \Tr\lbr Y_{A_1B_1} \Psym(d_1)\rbr \\
&= \Tr\lbr\tilde{\tau}_{A_1B_1A_2B_2}\Pasym(d_1\otimes 2)\rbr ,
\end{align*}
for the quantum state 
\[
\tilde{\tau}_{A_1B_1A_2B_2} = X_{A_1B_1}\otimes \frac{\Psym(2)}{\Tr\lbr \Psym(2)\rbr} + Y_{A_1B_1}\otimes \frac{\Pasym(2)}{\Tr\lbr \Pasym(2)\rbr} .
\]
Since
\[
\tilde{\tau}_{A_1A_2} = X_{A_1}\otimes \frac{\one_2}{2} + Y_{A_1}\otimes \frac{\one_2}{2} = \rho\otimes\frac{\one_2}{2} , 
\]
and
\[
\tilde{\tau}_{B_1B_2} = X_{B_1}\otimes \frac{\one_2}{2} + Y_{B_1}\otimes \frac{\one_2}{2} = \sigma\otimes\frac{\one_2}{2} ,
\] 
we have 
\[
T\lb \rho_1\otimes \rho_2 , \sigma_1\otimes \sigma_{2}\rb = \Tr\lbr\hat{\tau}_{A_1B_1A_2B_2}\Pasym(d_1\otimes d_2)\rbr \geq T\lb \rho\otimes \frac{\one_2}{2} , \sigma\otimes \frac{\one_2}{2}\rb .
\]
\end{proof}

Theorem \ref{thm:MaxMonViol} shows that the violation of monotonicity under partial traces is exhibited strongest when tensoring with the maximally mixed qubit state as in Theorem \ref{thm:ViolMon}.

\section{Stabilized quantum Wasserstein semi-distance}\label{sec:stab}

We have shown that the quantum optimal transport cost (and hence the quantum Wasserstein semi-distance) from above is not monotone under quantum channels and it is natural to ask whether this shortcoming can be fixed by modifying the definition. Let us define the following quantity:

\begin{defn}
For quantum states $\rho,\sigma\in \D\lb \C^d\rb$ we define the \emph{stabilized quantum optimal transport cost} as
\[
T_s(\rho,\sigma) = \inf_{\gamma} T\lb\rho\otimes \gamma,\sigma\otimes \gamma\rb ,
\]
where the infimum is over quantum states $\gamma\in \D\lb \C^D\rb$ of any dimensions $D\in\N$. The \emph{stabilized quantum Wasserstein semi-distance} of order $2$ is then defined as
\[
W_s(\rho,\sigma) = \sqrt{T_s(\rho,\sigma)} .
\]
\end{defn}

From Theorem \ref{thm:MaxMonViol} and its proof we immediately get a precise characterization of the stabilized quantum optimal transport cost in terms of the original quantum optimal transport cost and as a semidefinite program:

\begin{thm}\label{thm:FormulaStab}
For quantum states $\rho,\sigma\in \D\lb \C^d\rb$ we have
\begin{align*}
T_s\lb \rho,\sigma\rb &= T\lb \rho\otimes \frac{\one_2}{2}, \sigma\otimes \frac{\one_2}{2}\rb \\
&= \min_{X_{AB},Y_{AB}}\lb \Tr\lbr  X_{AB} \Psym(d) \rbr + \Tr\lbr  Y_{AB} \Pasym(d) \rbr\rb ,
\end{align*}
where the minimum is over positive operators $X_{AB},Y_{AB}\in \lb \M_d\otimes \M_d\rb^+$ such that 
\[
X_A + Y_A = \rho , \quad\text{ and }\quad X_B + Y_B = \sigma.
\]
\end{thm}

The stabilized quantum transport cost has many desirable properties:

\begin{lem}[Joint convexity]
For quantum states $\rho_1,\rho_2, \sigma_1, \sigma_2\in \D(\C^d)$ and $\lambda\in\lbr 0,1\rbr$ we have
\begin{align*}
T_s&(\lambda\rho_1+(1-\lambda)\rho_2,\lambda\sigma_1+(1-\lambda)\sigma_2)\\
&\quad\quad\quad \leq \lambda T_s(\rho_1,\sigma_1) + (1-\lambda)T_s(\rho_2,\sigma_2) . 
\end{align*}
\end{lem}
\begin{proof}
This follows easily from the joint convexity of $T$ (see \cite[Proposition 2.2]{cole2021quantum}) and Theorem \ref{thm:FormulaStab}. 
\end{proof}

We also have the following:

\begin{lem}[Invariance under tensor products]\label{thm:InvTP}
For quantum states $\rho,\sigma\in \D(\C^d)$ and any quantum state $\gamma\in \D(\C^D)$ we have
\[
T_s\lb \rho\otimes \gamma, \sigma\otimes \gamma\rb = T_s\lb \rho,\sigma\rb .
\]
\end{lem}

\begin{proof}
By Theorem \ref{thm:FormulaStab} we have 
\begin{align*}
T\lb \rho\otimes \gamma\otimes \frac{\one_2}{2}, \sigma\otimes\gamma\otimes \frac{\one_2}{2}\rb \geq T\lb \rho\otimes \frac{\one_2}{2}, \sigma\otimes \frac{\one_2}{2}\rb .
\end{align*}
But using Theorem \ref{thm:MonUnderTensor}, we also have
\[
T\lb \rho\otimes \gamma\otimes \frac{\one_2}{2}, \sigma\otimes\gamma\otimes \frac{\one_2}{2}\rb \leq T\lb \rho\otimes \frac{\one_2}{2}, \sigma\otimes \frac{\one_2}{2}\rb,
\]
Therefore, we have the equality 
\begin{align*}
T_s\lb \rho\otimes \gamma, \sigma\otimes \gamma\rb &= T\lb \rho\otimes \gamma\otimes \frac{\one_2}{2}, \sigma\otimes\gamma\otimes \frac{\one_2}{2}\rb \\
&= T\lb \rho\otimes \frac{\one_2}{2}, \sigma\otimes \frac{\one_2}{2}\rb = T_s(\rho,\sigma).
\end{align*}
\end{proof}

Finally, we note that the following lemma is clear by the definition of $T_s$:

\begin{lem}\label{lem:UnInv}
For quantum states $\rho,\sigma\in \D(\C^d)$ and any unitary $U\in \cU_d$ we have
\[
T_s\lb U\rho U^\dagger, U\sigma U^{\dagger}\rb = T_s\lb \rho,\sigma\rb .
\]
\end{lem}

We have the following lemma:

\begin{thm}
For any pair of quantum states $\rho,\sigma\in \D\lb \C^d\rb$ and any quantum channel $\phi:\M_{d}\ra \M_{d'}$ we have
\[
T_s\lb \phi(\rho),\phi(\sigma)\rb \leq T_s\lb \rho,\sigma\rb.
\]
\end{thm}

\begin{proof}
The proof follows a well-known strategy to show monotonicity of quantities under quantum channels (see, e.g.,~\cite{wolf2012quantum}). For any quantum channel $\phi:\M_{d}\ra \M_{d'}$ there exists a unitary $U\in \cU\lb \C^d \otimes \C^{d'}\otimes \C^{d'}\rb$ and a pure quantum state $\ket{\psi}\in \C^{d'}\otimes \C^{d'}$ such that 
\[
\phi(X) = \Tr_{E}\lbr U(X\otimes \proj{\psi}{\psi})U^\dagger\rbr ,
\]
for every $X\in \M_d$, where $\Tr_E$ denotes the trace over the first two tensor factors of $\C^{d}\otimes \C^{d'}\otimes \C^{d'}$. Furthermore note that the completely depolarizing channel $\Delta:\M_{dd'}\ra \M_{dd'}$ given by
\[
\Delta(X)=\Tr\lbr X\rbr\frac{\one_{dd'}}{dd'}.
\]
is a mixed unitary channel. 

It is easy to see that $T_s$ is monotone under mixed unitary channels since it is jointly convex and unitarily invariant. For any pair of quantum states $\rho,\sigma\in \D\lb \C^d\rb$ we have 
\begin{align*}
T_s &\lb \phi(\rho),\phi(\sigma)\rb =T_s\lb \frac{\one_{dd'}}{dd'}\otimes \phi(\rho), \frac{\one_{dd'}}{dd'}\otimes\phi(\sigma)\rb \\
&= T_s\lb (\Delta\otimes \ident_d)\lb U(\rho\otimes \proj{\psi}{\psi})U^\dagger\rb,(\Delta\otimes \ident_d)\lb U(\sigma\otimes \proj{\psi}{\psi})U^\dagger\rb\rb \\
&\leq T_s\lb U(\rho\otimes \proj{\psi}{\psi})U^\dagger, U(\sigma\otimes \proj{\psi}{\psi})U^\dagger\rb \\
& = T_s\lb \rho\otimes \proj{\psi}{\psi}, \sigma\otimes \proj{\psi}{\psi}\rb \\
&=T_s\lb \rho, \sigma\rb ,
\end{align*} 
where we used Theorem \ref{thm:InvTP} in the first equality, the definition of $\Delta$ in the second equality, and the fact that $\Delta$ is a mixed unitary channel in the inequality. Finally, we used unitary invariance (see Lemma~\ref{lem:UnInv}) and invariance under tensoring with pure quantum states, which is easily seen from the definition of the quantum optimal transport cost.
\end{proof}

\section{Conclusion}

We have shown that the quantum optimal transport cost and the quantum Wasserstein semi-distance introduced in~\cite{chakrabarti2019quantum} and studied in~\cite{cole2021quantum,friedland2022quantum,bistron2022monotonicity} are not monotone under the application of general quantum channels. To fix this shortcoming, we have proposed a stabilized version of the original definition. While this stabilized quantum optimal transport cost is monotone under quantum channels it is still open whether the corresponding stabilized quantum Wasserstein semi-distance satisfies the triangle inequality, which would turn it into a genuine distance. Clearly, this triangle inequality would follow if the triangle inequality for the original definition could be shown. 

\section*{Acknowledgments}
We thank Daniel Stilck Fran\c ca for interesting discussions and many comments that improved this article. AMH acknowledges funding from The Research Council of Norway (project 324944).

\bibliographystyle{alpha}
\bibliography{biblio.bib}

\end{document}